%% file: document.tex
\documentclass{article}

\usepackage[utf8]{inputenc}     % UTF 8
\usepackage{xspace} %space after a command
\usepackage{eurosym} %for the symbol euro

\usepackage[verbose,a4paper,portrait,twoside,twocolumn=false]{geometry}

\usepackage{calc} % for doing computations in LaTEX :)
\usepackage{url}
\usepackage{color}

\usepackage{tikz}
\usetikzlibrary{arrows,patterns,topaths}
\usepackage{xparse}

\usepackage{graphicx}

\usepackage{array} % extends environments «array» et «tabular»

\usepackage{amsmath}
\usepackage{amssymb}  % \mathbb
\DeclareMathAlphabet{\mathpzc}{OT1}{pzc}{m}{it} % \mathpzc
\usepackage{mathrsfs} % \mathscr

\usepackage{fancyvrb}
\usepackage[bookmarks={true},bookmarksopen={true}]{hyperref}
\usepackage[normalem]{ulem}
\usepackage{amsmath}

\input{macros.tex}

\begin{document}

\title{The Complexity of Tiling Problems}
\author{François Schwarzentruber \\ Univ Rennes, CNRS, IRISA, France}

\maketitle

\newcommand{\settilingtypes}{T}
\newcommand{\atile}{t}
\newcommand{\verticalconstraint}[2]{\tikz[scale=0.3]{
\node at (0.5, 0.5) {\ensuremath{#1}};
\node at (0.5, -0.5) {\ensuremath{#2}};
\draw (0, 0) rectangle (1, 1);
\draw (0, 0) rectangle (1, -1);
}}

\newcommand{\STA}[3]{STA(#1, #2, #3)}
\newcommand{\tilingproblem}[1]{TILING(#1)}

\begin{abstract}
In this document, we collected the most important complexity results of tilings. We also propose a definition of a so-called deterministic set of tile types, in order to capture deterministic classes without the notion of games. We also pinpoint tiling problems complete for respectively LOGSPACE and NLOGSPACE.
\end{abstract}

\section{Introduction}

As advocated by van der Boas \cite{van1997convenience}, tilings are convenient to prove lower complexity bounds. 
In this document, we show that tilings of finite rectangles with Wang tiles \cite{DBLP:journals/cacm/Wang60} enable to capture many standard complexity classes:

\begin{quote}
\FO (the class of decision problems defined by a first-order formula, see~\cite{immerman2012descriptive}),
\LOGSPACE, \NLOGSPACE, \PTIME, \NP, \PSPACE, \EXPTIME, \NEXPTIME, $k$-\EXPTIME and $k$-\EXPSPACE, for $k \geq 1$.
\end{quote}

  This document brings together many results of the literature. We recall some results from \cite{van1997convenience}. The setting is close to Tetravex \cite{DBLP:journals/ipl/TakenagaW06}, but the difference is that we allow a tile to be used several times. That is why we will the terminology \emph{tile types}. We also recall the results by Chlebus \cite{DBLP:journals/jcss/Chlebus86} on tiling games, but we simplify the framework since we suppose that players alternate at each row (and not at each time they put a tile).

The first contribution consists in capturing deterministic time classes with an existence of a tiling, and without any game notions. We identify a syntactic class of set of tiles, called  \emph{deterministic} set of tiles. For this, we have slightly adapted the definition of the encoding of executions of Turing machines given in \cite{van1997convenience}.

The second contribution is the connection between one-dimensional tilings and the classes \LOGSPACE and \NLOGSPACE. In particular, we rely on the fact that the reachability problem in directed graphs is \NLOGSPACE-complete \cite{Papadimitriou}, and the reachability problem in undirected graphs is in \LOGSPACE \cite{DBLP:journals/jacm/Reingold08}. 
There are small differences compared to the literature. Note that Grädel (see \cite{DBLP:journals/siamcomp/Gradel90}, p. 800 before Th. 7.1) also introduced one-dimensional tilings, more precisely domino games. According to Grädel \cite{DBLP:journals/siamcomp/Gradel90} (Remark p. 802), they were also introduced by J. Tor\`an in his PhD thesis. Note that we generalize some result of  Etzion-Petruschka et al. who also considered one-dimensional tilings (see Th. 2.7 in \cite{DBLP:journals/tcs/Etzion-PetruschkaHM94}).

\emph{Outline. } First we recall basic definitions about tilings in Section~\ref{section:definition}. Then, we recall results about existence of tilings and classes above \PTIME in~\ref{section:existenceaboveP}. We then continue with tiling games in Section~\ref{section:games}. We then show how to get rid off games in order to capture deterministic classes in Section~\ref{section:deterministic}. We finish with existence of tilings and classes below \PTIME in Section~\ref{section:existencebelowP}.

\newcommand{\tileseed}{t_0}

\section{Basic Definitions}
\label{section:definition}

A tile type $t$ specifies the four colors of a tile in the left, up, right, down directions.
%, e.g:
%
%\begin{center}
%	\begin{tikzpicture}
%	\tile{0}{0}{\tileyellow}{\tilered}{\tilegreen}{\tilered}
%	%\tile{2}{0}{\tilered}{\tilered}{\tilered}{\tilered}
%	\end{tikzpicture}
%\end{center}
\newcommand{\colors}{\mathcal C}
Formally, let $\colors$ be a countable set of \emph{colors}; a special color being \emph{white}. A tile type $t$ is an element of $\colors^4$, written $\tuple{left(t), up(t), right(t), down(t)}$. 

\newcommand{\tilingfunction}{\tau}
\newcommand{\whitetext}{white}

Let $\settilingtypes$ be a finite set of tile types.
A $\settilingtypes$-\emph{tiling} of the finite $H \times W$ rectangle is a function $\tilingfunction : \set{1, \dots, H} \times \set{1, \dots, W} \rightarrow \settilingtypes$ such that:
\begin{enumerate}
	\item $left(\tilingfunction(i, 1)) = right(\tilingfunction(i, W)) = \whitetext$ \hfill for all $i \in \set{1, \dots, H}$;
	\item $top(\tilingfunction(1, j)) = bottom(\tilingfunction(H, j)) = \whitetext$ \hfill for all $j \in \set{1, \dots, W}$;
	\item $right(\tilingfunction(i, j)) = left(\tilingfunction(i, j{+}1))$
	\hfill  for all $i \in \set{1, \dots, H}$, $j \in \set{1, \dots, W{-}1}$;
	\item $bottom(\tilingfunction(i, j)) = top(\tilingfunction(i{+}1, j))$
	\hfill  for all $i \in \set{1, \dots, H{-}1}$,  $j \in \set{1, \dots, W}$.
\end{enumerate} 
Constraint 1 means that the left of the left-most tiles and the right of the right-most tiles should be white\footnote{As you will see, this constraint is important to identify the beginning of the tape of a Turing machine and to avoid the head to disappear when its head is the left-most cell by triggering a transition moving the head to the left. This constraint will also help to deterministic set of tiles.}. Constraint 2 says that the top of the top-most tiles and the bottom of the bottom-most tiles should be white\footnote{This constraint will also help to deterministic set of tiles.}. Constraint 3 corresponds to the horizontal constraint and constraint 4 to the vertical constraint. %Note that  no constraint at the bottom of the rectangle.

	Let $w: \ensN \rightarrow \ensN$ and $h: \ensN \rightarrow \ensN$.
We aim to tile the finite $h(n) \times w(n)$ rectangle, as shown in Figure~\ref{figure:tiling}. The top-left tile $\tileseed$ plays the role of a seed and is given. 

\begin{figure}[t]
	\begin{center}
				\begin{tikzpicture}[scale=0.7]
				\draw[fill=lightgray] (0, 3) rectangle (3, 0);
%				\foreach \x in {0.5, 1.5} 
%				\foreach \y in {0.5, 2.5}{
%					\draw[->] (\x, \y) -- (\x+1, \y);
%				}
%				\foreach \x in {0.5, 1.5} 
%				\foreach \y in {1.5}{
%					\draw[->] (\x+1, \y) -- (\x, \y);
%				}		
%				\draw[->] (2.5, 2.5) edge[in=0, out=0] (2.5, 1.5);
%				\draw[->] (0.5, 1.5) edge[in=180, out=180] (0.5, 0.5);
				\tile{0}{2}{\tilewhite}{\tilewhite}{\tilegreen}{\tilered}
				%		\tile{1}{2}{\tilegreen}{\tilered}{\tilegreen}{\tileyellow}
				%		\tile{2}{2}{\tilegreen}{\tilered}{\tilegreen}{\tileyellow}
				
				\end{tikzpicture}
				\hfill
		\begin{tikzpicture}[scale=0.7]
%		\node at (-1, 0.5) {$\exists$};
%		\node at (-1, -0.5) {$\forall$};
%		\node at (-1, -1.5) {$\exists$};
		\tile 0 0 \tilewhite \tilewhite \tilegreen\tilered
		\tile 1 0 \tilegreen\tilewhite\tilegreen\tileyellow
		\tile 2 0 \tilegreen\tilewhite\tilewhite\tileyellow
		\tile 0 {-1} \tilewhite \tilered\tilered\tilered
		\tile 1 {-1} \tilered \tileyellow\tilered\tilegreen
		\tile 2 {-1} \tilered \tileyellow\tilewhite\tileyellow
		\tile 0 {-2} \tilewhite\tilered\tilegreen\tilewhite
		\tile 1 {-2} \tilegreen\tilegreen\tilered\tilewhite
		\tile 2 {-2} \tilered	\tileyellow\tilewhite\tilewhite
		\end{tikzpicture}
	\end{center}
	\caption{Rectangle to be tiled and a solution.\label{figure:tiling}}
\end{figure}
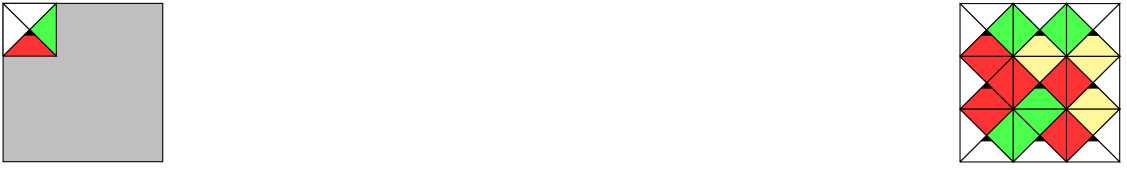

\begin{definition}
	Let $h: \ensN \rightarrow \ensN$ and $w: \ensN \rightarrow \ensN$.
	$\tilingproblem{h, w}$ is the following decision problem:
	\begin{itemize}
\item input: an integer $n$ given in \textbf{unary}, a finite set $\settilingtypes$ of tiling types, a tile $\tileseed$;
\item output: yes, if there is a $\settilingtypes$-tiling $\tilingfunction$ of the $h(n)\times w(n)$ rectangle such that $\tilingfunction(1, 1) = \tileseed$; no, otherwise.

% a winning strategy for player $\exists$ to the game described below, in the $h(n) \times w(n)$ rectangle, using $\tileseed$ as a seed, and respecting the abstract sequence $\playerabstractsequence$ of players, no otherwise.
	\end{itemize}
\end{definition}

We write directly the expression $w(n)$ instead of the function $w$. For instance, we write $2^n$ instead of $n \mapsto 2^n$. Same for $h$.
We also consider the variant in which the height is arbitrary.

\begin{definition}
	Let  $w: \ensN \rightarrow \ensN$.
	$\tilingproblem{*, w}$ is the following decision problem:
	\begin{itemize}
		\item input: an integer $n$ given in \textbf{unary}, a finite set $\settilingtypes$ of tiling types, a tile $\tileseed$;
		\item output: yes, if there are an integer $h$ and a $\settilingtypes$-tiling $\tilingfunction$ of the $h \times w(n)$ rectangle such that $\tilingfunction(1, 1) = \tileseed$; no, otherwise.
		
		% a winning strategy for player $\exists$ to the game described below, in the $h(n) \times w(n)$ rectangle, using $\tileseed$ as a seed, and respecting the abstract sequence $\playerabstractsequence$ of players, no otherwise.
	\end{itemize}
\end{definition}

\section{Existence of Tilings and Classes above \PTIME}
\label{section:existenceaboveP}
\newcommand{\tmstatesfirst}{Q}
\newcommand{\tmstatessecond}{Q'}
\newcommand{\tminitialstate}{q_0'}
\newcommand{\tmfinalstate}{q_{\mathsf{f}}}
\newcommand\turingmachine{M}
\subsection{Encoding  Executions of Turing Machines}

In this section, we explain how to encode an execution of a Turing machine as a tiling. We slightly adapt the \emph{normalization of Turing machines} given in \cite{van1997convenience}, especially for being able to capture deterministic tilings (see Section~\ref{section:deterministic}). As advocated in \cite{van1997convenience}, normalization does not impact on the complexity classes.
Without loss of generality, we suppose that the machine is \emph{normalized}.

\begin{definition}
	\label{definition:normalizedturingmachine}
	A Turing machine $\turingmachine$ is \emph{normalized} if its set of states is partitioned in two disjoint subsets $\tmstatesfirst$ and $\tmstatessecond$ (see Figure~\ref{figure:tmstates})such that:
	\begin{enumerate}
		\item  the initial state $\tminitialstate$ is in $\tmstatessecond$;
		\item transitions going out from $\tmstatesfirst$ go in $\tmstatessecond$ and makes the cursor move right or makes the cursor stay at its current position;
		\item  transitions going out from $\tmstatessecond$ go in $\tmstatesfirst$ and makes the cursor move left  or makes the cursor stay at its current position;
		\item the final (accepting) state $\tmfinalstate$ is in $\tmstatesfirst$;
		\item if the machine reaches the final accepting state $\tmfinalstate$, then the tape has already been erased (all cells contain the blank symbol \textvisiblespace).
		\item the final (accepting) state $\tmfinalstate$ is in $\tmstatesfirst$ and has a copy in $\tmfinalstate'$ in $\tmstatessecond$, there are transitions between them, that do not move the cursor, do not change the tape.
	\end{enumerate} 
\end{definition}

\begin{figure}
	\begin{center}
		\begin{tikzpicture}[yscale=0.8]
			\draw[rounded corners=2mm] (0, -1) rectangle (12, 5);
			\node[draw, minimum width=3cm, minimum height=3cm, rounded corners=2mm] (Q) at (2, 2) {};
			\node[draw, minimum width=3cm, minimum height=3cm, rounded corners=2mm] (Q') at (10, 2) {};
			\node at (10, 3) (q0) {$\tminitialstate$};
			\node at (1, 3)  {$\tmstatesfirst$};
			\node at (2.5, 2) (finalstate) {$\tmfinalstate$};
			\node at (9.5, 2) (finalstate') {$\tmfinalstate'$};
			\node at (11, 3)  {$\tmstatessecond$};
			
			%\draw[->] (finalstate) edge[loop, looseness=4] node[above] {loop} (finalstate);
			\draw[->] (finalstate) edge[bend left=5] (finalstate');
			\draw[->] (finalstate') edge[bend left=5] (finalstate);
			\draw (Q) edge[->,  bend left=20] node [above, text width=4cm] {some transitions, the head is going to the right or not moving} (Q');
			\draw (Q') edge[->,  bend left=20] node [below, text width=4cm] {some transitions, the head is going to the left or not moving} (Q);			
		\end{tikzpicture}
	\end{center}
	\caption{Set of states of a normalized Turing machine, partitioned in $\tmstatesfirst$ and~$\tmstatessecond$.\label{figure:tmstates}}
\end{figure}
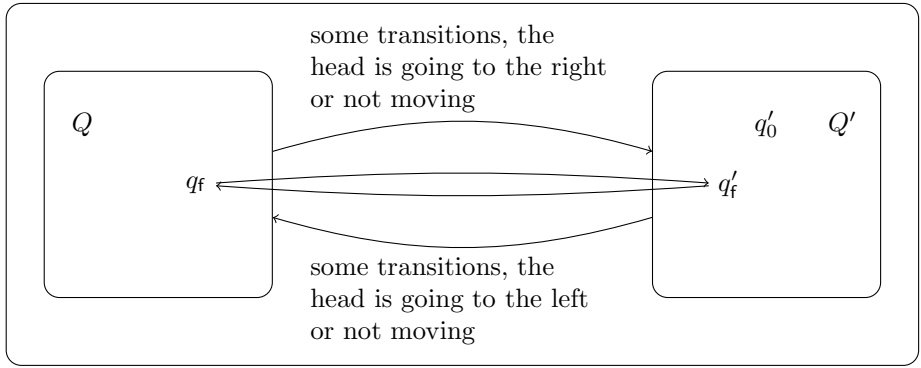

We encode an execution almost as in \cite{van1997convenience}. Let us consider a Turing machine $\turingmachine$ and an input word $\word$. Figure~\ref{figure:tiletypesfortm} shows the set of tiles $\settilingtypes_{\turingmachine, \word}$. These tiles enable to represent any execution of length $H$ of $\turingmachine$ on input $\word$, that uses at most $W$ cells, with a tiling of the $H \times W$-rectangle. 
 The idea is that we always alternate between $\tmstatesfirst$ and $\tmstatessecond$. Being in a state in $\tmstatesfirst$ (resp. in $\tmstatessecond$) is tagged at any tile in a row with the absence (pres. presence) of the symbol ' (prime)\footnote{This difference will help to define deterministic set of tiles, see Section~\ref{section:deterministic}}. 
States in $\tmstatesfirst$ are noted $q$ etc. States in $\tmstatessecond$ are noted $q'$, etc. The color $a'$ is copy of the symbol $a$, it is used to keep track on the full row  whether the current state is in $\tmstatesfirst$ or $\tmstatessecond$. 
Figure~\ref{figure:tmencoding} shows an example of such an encoding of an execution of a machine on the input word~$bba$.
	
The machine being normalized prevents to have two adjacent tiles that would create two cursor positions

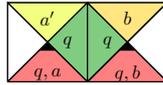
\begin{center}
	\scalebox{0.7}{
	\begin{tikzpicture}[scale=1.5]

	\tiletextcolor {3} {0} {} {a'} {q} {q,a}
	\tiletextcolor {4} {0} {q} {b} {} {q,b}
	\end{tikzpicture}}
\end{center}

because we would have allowed to enter a state $q$ both with a transition moving the cursor to the left and with another transition moving the cursor to the right.

When the machines reaches $\tmfinalstate$ it runs forever, and the tiling finishes with a line of white at the bottom. If not, either it runs forever or it gets stuck; there is no line of white at the bottom in the tiling corresponding to the execution.
We could have simply assumed that we loop in $\tmfinalstate$ but the notion of deterministic set of tile types would have been more difficult to define (see Section~\ref{section:deterministic}).

\begin{figure}
	\begin{center}
		\scalebox{0.7}{
		\begin{tikzpicture}[scale=1.5]
		\tiletextcolor {0} {0} {} {a} {} {a'}
		
		\tiletextcolor {1.5} {0} {} {a'} {} {a}

		\tiletextcolor {3} {0} {q'} {a} {} {q',a'}
		
		\tiletextcolor {4.5} {0} {} {a'} {q} {q,a}
		
			\node at (4.5, -0.3){for all symbols $a$, for all $q \in \tmstatesfirst$, $q' \in \tmstatessecond$};
		
		\tiletextcolor {0} {-1.5} {} {q,a} {q'} {b'}
		\node at (3, -1.5+0.5){for al transitions $(q, a, b, \rightarrow, q')$};
		
		\tiletextcolor {0} {-3} {q} {q',a'} {} {b}
		\node at (3, -3+0.5){for all transitions $(q', a, b, \leftarrow, q)$};
		
		\tiletextcolor {0} {-4.5} {} {q,a} {} {q',b'}
		
		\tiletextcolor {1.5} {-4.5} {} {q',a'} {} {q,b}
		
		\node at (4.5, -4.5+0.5){for all transitions $(q, a, b, \cdot, q')$};
		
		\draw[fill=gray!30!white, draw=none] (-1, -4.8) rectangle (1.1, -6.2);
		
		\node at (-0.5, -5.5) {$\tileseed := $};
		\tiletextcolor {0} {-6} {} {} {2} {q_0',w_1'}
		\tiletextcolor {1.2} {-6} {2} {} {3} {w_2'}
		\tiletextcolor {4} {-6} {|w|} {} {} {w_{|w|}'}
		\node at (3, -6) {$\dots$};
		\tiletext {5.5} {-6} {} {} {} {'}

		\tiletext {0} {-7.5} {} {\textvisiblespace'} {} {}

		\tiletext {1.5} {-7.5} {} {$\tmfinalstate'$,\textvisiblespace'} {} {}
		
		\end{tikzpicture}}

	\end{center}
	\caption{Set of tile types for encoding an execution of a given Turing machine~$\turingmachine$ on input word $w$.\label{figure:tiletypesfortm}}
\end{figure}
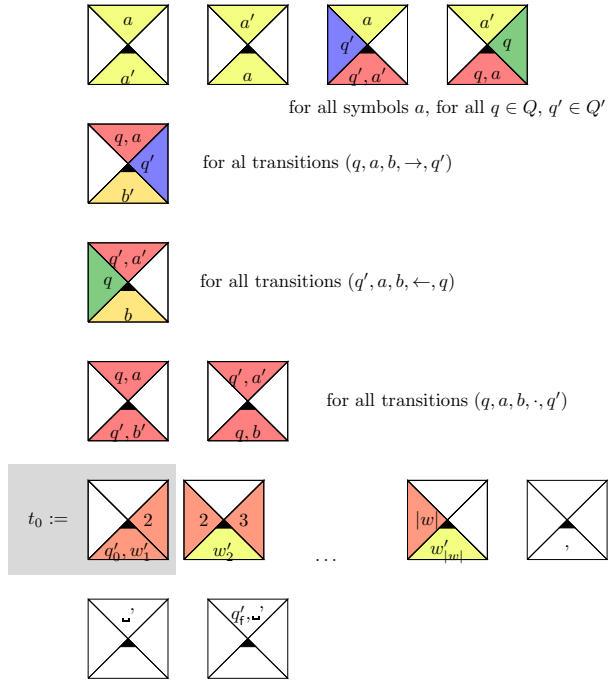

	\begin{figure}
		\begin{center}
			\scalebox{0.7}{
			\begin{tikzpicture}[scale=1.5]
			\tiletextcolor 0 1  {} {} {2}  {q_0',b'}
			\tiletextcolor 1 1 {2} {} {3} {b'}
			\tiletextcolor 2 1 {3} {} {}  {a'}
			\tilecopytoprime 3 1 {}
			\tilecopytoprime 4 1 {}
			\tilecopytoprime 5 1 {}
			\tilecopytoprime 6 1 {}	
			
			\tiletextcolor 0 {0} {} {q_0',b'} {} {q_1,a}
			\tilecopyfromprime 1 {0} {b}
			\tilecopytoprime 2 {0} {a}
			\tilecopyfromprime 3 {0} {}
			\tilecopyfromprime 4 {0} {}
			\tilecopyfromprime 5 {0} {}
			\tilecopyfromprime 6 {0} {}

			\tilegotoright 0 {-1} {q_1} {a} {a'} {q_2'}
			\tilerecfromright 1 {-1} {q_2'} {b'}
			\tilecopytoprime 2 {-1} {a}
			\tilecopytoprime 3 {-1} {}
			\tilecopytoprime 4 {-1} {}
			\tilecopytoprime 5 {-1} {}
			\tilecopytoprime 6 {-1} {}
			
			\tilecopyfromprime 0 {-2} {a}
			\tiletextcolor 1 {-2} {} {q_2',b'} {} {q_2,a}
			\tilecopyfromprime 2 {-2} {a}
			\tilecopyfromprime 3 {-2} {}
			\tilecopyfromprime 4 {-2} {}
			\tilecopyfromprime 5 {-2} {}
			\tilecopyfromprime 6 {-2} {}
%			
%			\tilecopy 0 {-3} {a}
%			\tilecopy 1 {-3} {b}
%			\tilegotoright 2 {-3} {q_1} {a} {a} {q_0}
%			\tilerecfromright 3 {-3} {q_0} {}
%			\tilecopy 4 {-3} {}
%			\tilecopy 5 {-3} {}
%			\tilecopy 6 {-3} {}
%			
%			\tilecopy 0 {-4} {a}
%			\tilecopy 1 {-4} {b}
%			\tilecopy 2 {-4} a
%			\tilegoto 3 {-4} {q_0} {} {} {q_2}
%			\tilecopy 4 {-4} {}
%			\tilecopy 5 {-4} {}
%			\tilecopy 6 {-4} {}
%			
%			\tilecopy 0 {-5} {a}
%			\tilecopy 1 {-5} {b}
%			\tilerecfromleft 2 {-5} {q_f} a
%			\tilegotoleft 3 {-5} {q_2} {} {} {q_f}
%			\tilecopy 4 {-5} {}
%			\tilecopy 5 {-5} {}
%			\tilecopy 6 {-5} {}
%			
%			\tilecopy 0 {-6} {a}
%			\tilecopy 1 {-6} {b}
%			\tilecopy 2 {-6} {q_f,a}
%			\tilecopy 3 {-6} {}
%			\tilecopy 4 {-6} {}
%			\tilecopy 5 {-6} {}
%			\tilecopy 6 {-6} {}
			
			\end{tikzpicture}}
			
			$\vdots$
		\end{center}
		\caption{Encoding of an execution in a tiling of an execution on the input word $bba$.\label{figure:tmencoding}}
	\end{figure}
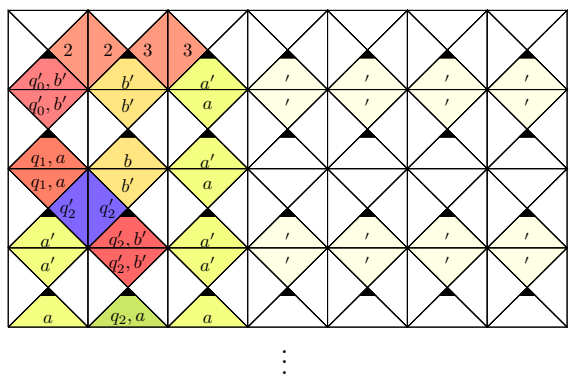

	\subsection{Existence of Tilings in Squares}
	
First we tackle $\tilingproblem{n, n}$. Some readers may be surprised by the relevance of that problem, in which $n$ is given in unary. That assumption is quite natural: any tiling requires $\Omega(n^2)$ memory cells to be stored; memory cells you need to allocate anyway to store that tiling. This is close to the assumption made in bounded planning (called polynomial-length planning problem), for which the bound is also written in unary (see~\cite{DBLP:conf/jelia/Turner02}).

	\begin{theorem}
		\label{theorem:tilingnnNP-complete}
		$\tilingproblem{n, n}$ is \NP-complete.
	\end{theorem}
	
	\begin{proof}
A non-deterministic algorithm deciding $\tilingproblem{n, n}$ in polynomial-time consists in guessing a function $\tilingfunction$ and checking that $\tilingfunction$ is indeed a tiling of the $n \times n$-rectangle, and that $\tilingfunction(1, 1) = \tileseed$.

Let $A$ be a problem in $\NP$. There exists a non-deterministic Turing machine $\turingmachine$ that decides $A$ in polynomial-time. W.l.o.g. we suppose that the machine is normalized (see Definition~\ref{definition:normalizedturingmachine}, and that there is a polynomial $f$ such that any execution on an input $w$ of size $|w|$, either stops in strictly less than $f(|w|)$ steps, or reaches $\tmfinalstate$ in strictly less than $f(|w|)$ steps and keeps running forever.

$A$ reduces to $\tilingproblem{n, n}$ in polynomial-time, even in log-space: the reduction is $tr(w) = \tuple{\settilingtypes_{\turingmachine, \word}, n, \tileseed}$ where $n := f(|w|)$, $\settilingtypes_{\turingmachine, \word}$ and $\tileseed$ are shown in Figure~\ref{figure:tiletypesfortm}.

We have $w \in A$ iff $tr(w) \in \tilingproblem{n, n}$. \fbox{$\Leftarrow$} If $w \in A$, then there is an accepting execution. Thus, we can tile the $n \times n$-rectangle using that execution as shown in Figure~\ref{figure:tmencoding}. \fbox{$\Rightarrow$} If there is a tiling of the $n \times n$-rectangle, then the seed enforces the first row to contain the input word. The other tiles enforce the tiling to represent an execution. As the bottoms of the bottom-most tiles are white, it means that the execution reaches $\tmfinalstate$. So the execution is accepting and $w \in A$.
	\end{proof}

Note that we could define the variant of $\tilingproblem{n, n}$ in which no seed $\tileseed$ is given in the input. The problem is to tile the $n\times n$-rectangle without the seed constraint. This problem is called to be the seed-free variant.

	\begin{theorem}
The seed-free variant		$\tilingproblem{n, n}$ is \NP-complete.
	\end{theorem}
	
	\begin{proof}
		For the \NP-hardness of the seed-free variant, it suffices to add "numbers" in colors in order to count.
	\end{proof}

	\newcommand{\expo}[2]{exp_{#1}(#2)}

In Theorem~\ref{theorem:tilingnnNP-complete}, the size of the square is $n$. If the size becomes exponential in $n$, double-exponential in $n$, etc., we capture the class $\NEXPTIME$, $2\NEXPTIME$, etc. That is why we define $\expo k n$ inductively on $k$:
\begin{itemize}
	\item $\expo 0 n := n$;
	\item $\expo k n := 2^{\expo {k-1} n}$ for all $k \geq 1$.
\end{itemize} 

In other words, $\expo k n$ is
$$ \begin{matrix}
\underbrace{2_{}^{2^{{}^{.\,^{.\,^{.\,^n}}}}}}.\\
\qquad\quad\ \ \ k\mbox{ occurrences of }2
\end{matrix}$$
%
%
%In the same way, we obtain:
	\begin{theorem}
		\label{theorem:knexptime}
		$\tilingproblem{\expo k n, \expo k n}$ is k\NEXPTIME-complete.
	\end{theorem}

	\subsection{Existence of Tilings in Rectangles of Arbitrary Height}

\begin{theorem}
	\label{theorem:PSPACE}
	$\tilingproblem{ 2^n, n}$ and $\tilingproblem{*, n}$ are \PSPACE-complete.
\end{theorem}

\begin{proof}
	A non-deterministic algorithm deciding $\tilingproblem{2^n, n}$ that runs in \linebreak[4] polynomial-space consists in guessing the tiling on row by row. We store the previous row, the current row and the $n$-bit index of the current row. For $\tilingproblem{*, n}$, we just do not care about the index of the current row.

	Let $A$ be a problem in $\PSPACE$.  There exists a machine $\turingmachine$ that decides $A$. W.l.o.g. we suppose that the machine is normalized (see Definition~\ref{definition:normalizedturingmachine}, and that there is a polynomial $f$ such that any execution on an input $w$ of size $|w|$ uses at most $f(|w|)$ cells and that, either stops in strictly less than $2^{f(|w|)}$ steps, or reaches $\stateaccept$ in strictly less than $2^{f(|w|)}$ steps and keeps running forever.
	
	The reduction is the same than in the proof of Theorem~\ref{theorem:tilingnnNP-complete}.
\end{proof}

In the same way, we obtain:
\begin{theorem}
	\label{theorem:kEXPSPACE}
	Let $k \geq 1$.
	$\tilingproblem{\expo {k+1} n, \expo {k} n}$ and 	$\tilingproblem{*, \expo {k} n}$ are  k\EXPSPACE-complete.
\end{theorem}

\section{Two-player Games}
\label{section:games}

In order to capture alternating classes \cite{chandra1976alternation}, we introduce two players: $\exists$ and $\forall$. Each row is owned by some player. Each move consists in adding a row below the current one, by choosing tiles among a given finite set of tile types $\settilingtypes$, so the colors match. Figure~\ref{figure:tilingalternation} shows a finished tiling  game: player $\exists$ chose the first row, then player $\forall$ chose the second row and player $\exists$ chose the third row.

\begin{figure}[!h]
	\begin{center}
%		\begin{tikzpicture}
%		\draw[fill=lightgray] (0, 3) rectangle (3, 0);
%		\foreach \x in {0.5, 1.5} 
%		\foreach \y in {0.5, 2.5}{
%			\draw[->] (\x, \y) -- (\x+1, \y);
%		}
%		\foreach \x in {0.5, 1.5} 
%		\foreach \y in {1.5}{
%			\draw[->] (\x+1, \y) -- (\x, \y);
%		}		
%		\draw[->] (2.5, 2.5) edge[in=0, out=0] (2.5, 1.5);
%		\draw[->] (0.5, 1.5) edge[in=180, out=180] (0.5, 0.5);
%		\tile{0}{2}{\tilewhite}{\tilewhite}{\tilegreen}{\tilered}
%		%		\tile{1}{2}{\tilegreen}{\tilered}{\tilegreen}{\tileyellow}
%		%		\tile{2}{2}{\tilegreen}{\tilered}{\tilegreen}{\tileyellow}
%		
%		\end{tikzpicture}
%		\hfill
		\begin{tikzpicture}[scale=0.7]
		\node at (-1, 0.5) {$\exists$};
		\node at (-1, -0.5) {$\forall$};
		\node at (-1, -1.5) {$\exists$};
		\tile 0 0 \tilewhite \tilewhite \tilegreen\tilered
		\tile 1 0 \tilegreen\tilewhite\tilegreen\tileyellow
		\tile 2 0 \tilegreen\tilewhite\tilewhite\tileyellow
		\tile 0 {-1} \tilewhite \tilered\tilered\tilered
		\tile 1 {-1} \tilered \tileyellow\tilered\tilegreen
		\tile 2 {-1} \tilered \tileyellow\tilewhite\tileyellow
		\tile 0 {-2} \tilewhite\tilered\tilegreen\tilewhite
		\tile 1 {-2} \tilegreen\tilegreen\tilered\tilewhite
		\tile 2 {-2} \tilered	\tileyellow\tilewhite\tilewhite
		\end{tikzpicture}
	\end{center}
	\caption{Finished tiling game, we suppose players alternate at each row.\label{figure:tilingalternation}}
\end{figure}
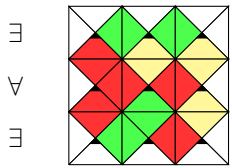

\subsection{Definition}
\newcommand{\playerabstractsequence}{\mathsf{plSeq}}
The ownership of rows is described by an abstract sequence $\playerabstractsequence$. For instance, if $\playerabstractsequence$ is $\exists^*$, it means that all rows belong to player $\exists$. If $\playerabstractsequence$ is $(\exists \forall)^*$, it means that the first, third... rows belong to player $\exists$ while the second, fourth... rows belong to player $\forall$. We will not develop a full theory of abstract sequences, since we will only use simple patterns. Player $\exists$ wins if the rectangle is fully tiled.

\begin{definition}
	Given $h: \ensN \rightarrow \ensN$ and $w: \ensN \rightarrow \ensN$, and an abstract sequence $\playerabstractsequence$, we define $\tilingproblem{h(n), w(n), \playerabstractsequence}$ to be the following decision problem:

	\begin{itemize}
		\item input: an integer $n$ given in unary, a finite set $\settilingtypes$ of tiling types, a tile $\tileseed$;
		\item Yes, if there is a winning strategy for player $\exists$ to the game described below, in the $h(n) \times w(n)$ rectangle, using $\tileseed$ as a seed, and respecting the abstract sequence $\playerabstractsequence$ of players; no otherwise.
	\end{itemize}
	
\end{definition}

Remark that $\tilingproblem{h(n), w(n)}$ is $\tilingproblem{h(n), w(n), \exists^*}$.

\subsection{Complexity Results}

Proofs are fastidious but, if players alternate, we capture alternating classes~\cite{chandra1976alternation}, and thus deterministic classes via \APTIME = \PSPACE and A$k$EXPTIME = $k$\EXPSPACE.
\begin{theorem}
	\label{theorem:tilingnnPSPACEgame}
	$\tilingproblem{n, n, (\exists\forall)^*}$ is \PSPACE-complete.
\end{theorem}

\begin{theorem}
	\label{theorem:gameskexpspace}
	$\tilingproblem{\expo k n, \expo k n, (\exists\forall)^*}$ is $k$\EXPSPACE-complete.
\end{theorem}

In the same way, as $\APSPACE= \EXPTIME$ and $\kAEXPSPACE[k] = \kEXPTIME[k]$.
\begin{theorem}
	$\tilingproblem{2^n, n, (\exists\forall)^*}$ and $\tilingproblem{*, n, (\exists\forall)^*}$ are \EXPTIME-complete.
\end{theorem}

\begin{theorem}
	\label{theorem:existencekexp}
	Let $k \geq 1$. 
	$\tilingproblem{\expo k n, \expo {k-1} n,  (\exists\forall)^*}$ and 	$\tilingproblem{*, \expo {k-1} n,  (\exists\forall)^*}$ are $k$\EXPTIME-complete.
\end{theorem}

The polynomial hierarchy is captured as follows.
\begin{theorem}
	\label{theorem:PH}
	Let $k \geq 1$.
	
	\begin{itemize}
		\item 	$\tilingproblem{kn, n, \exists^n(\forall^n\exists^n)^{k-1}}$ is $\Sigma_k^P$-complete;
		\item 	$\tilingproblem{kn, n,  \forall^n(\exists^n\forall^n)^{k-1}}$ is $\Pi_k^P$-complete.
	\end{itemize}

\end{theorem}

Interestingly, we can capture the exotic class \AEXPTIMEpol (see \cite{DBLP:conf/jelia/BozzelliDP12} for instance), the class of problems decided by an alternating Turing machine in exponential time but with a polynomial number of alternations. Our reformulation is very closed from the problem called \emph{multi-tiling problem} introduced in \cite{DBLP:journals/corr/abs-1709-02094} that consists in tiling several $2^n \times 2^n$-squares. That tiling problem is used in~\cite{LICS2019DEMRI}.

\begin{theorem}
	\label{theorem:aexppol}
	$\tilingproblem{2 \times n \times 2^n, 2^n, (\exists^{2^n}\forall^{2^n})^*}$ is \AEXPTIMEpol-complete.
\end{theorem}

\begin{proof}
	Let $A$ be a problem in \AEXPTIMEpol.
	 There is a alternating Turing machine deciding $A$ in exponential time, with at most a polynomial number of alternation. As mentioned in \cite{DBLP:conf/atal/CharrierS15}, we can suppose that player $\exists$ 
	plays first, that each portion of the execution played by the player $\exists$ and each portion of the execution played by the player $\forall$ are of the same length
	 $2^{f(|w|)}$ where $f$ is a polynomial and $w$ is the input word. We suppose that there are $2\times a(|w|)$ such portions. W.l.o.g, we suppose that $a(|w|) = f(|w|)$. We furthermore suppose that the machine is normalized.
	
	The reduction is $tr(w) = \tuple{\settilingtypes_{\turingmachine, \word}, n, \tileseed}$ 
	where $\settilingtypes_{\turingmachine, \word}, \tileseed$ are given in Figure~\ref{figure:tmencoding} and $n := f(|w|)$.
	
\end{proof}

\section{Deterministic Tilings}
\label{section:deterministic}

\newcommand{\deterministictilingproblem}[1]{det\tilingproblem{#1}}

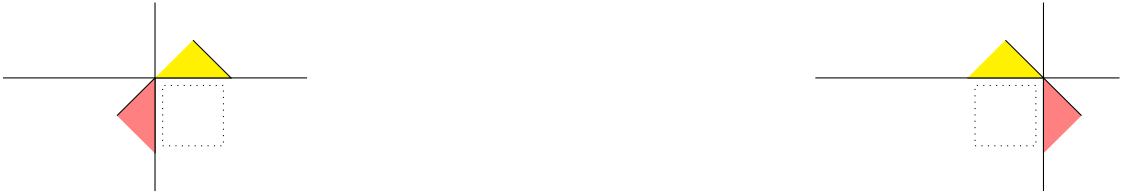
\begin{figure}[h]
	\begin{tikzpicture}
	\draw (-2, 1) -- (2, 1);
	\draw (0, -0.5) -- (0,2);
	\draw[fill=red!50] (0, 0) -- (0, 1) -- (-0.5, 0.5);
	\draw[fill=yellow] (0, 1) -- (1, 1) -- (0.5, 1.5);
	
	\draw[dotted] (0.1, 0.1) -- (0.9, 0.1) -- (0.9, 0.9) -- (0.1, 0.9) -- cycle;
	
	\end{tikzpicture}
	\hfill
		\begin{tikzpicture}
	\draw (-2, 1) -- (2, 1);
	\draw (1, -0.5) -- (1, 2);
	\draw[fill=red!50] (1, 0) -- (1, 1) -- (1.5, 0.5);
	\draw[fill=yellow] (0, 1) -- (1, 1) -- (0.5, 1.5);
	
	\draw[dotted] (0.1, 0.1) -- (0.9, 0.1) -- (0.9, 0.9) -- (0.1, 0.9) -- cycle;
	
	\end{tikzpicture}
	\caption{Deterministic set of tiles: there is a \emph{at most one} tile from $\settilingtypes$ that fits in the dotted square.\label{figure:deterministictiles}}
\end{figure}

In order to capture deterministic classes without games (no alternation between player $\exists$ and $\forall$), we introduce the notion of a deterministic set of tiles. 

\subsection{Deterministic Set of Tiles}

The idea is that a set $\settilingtypes$ of tiles is said to be \emph{deterministic} if there is at most one tile to complete a tiling, as shown in Figure~\ref{figure:deterministictiles} -- the direction depends on the top color. More precisely:

\newcommand{\colorsone}{Col}
\newcommand{\colorstwo}{Col'}
\begin{definition}
A set $\settilingtypes$ of tiles is deterministic if there is a partition  $\colors := \colorsone \sqcup \colorstwo$ such that $\whitetext \in \colorsone$ and:
\begin{itemize}
	\item for all tiles $t \in \settilingtypes$, $top(t) \in \colorsone$ iff $bottom(t) \in \colorstwo$;
	\item for all colors $c \in \colorsone$, for all colors $c' \in \colors$, there is at most one element $t$ such that $left(t) = c$ and $top(t) = c'$;
	\item for all colors $c \in \colorstwo$, for all color $c' \in \colors$, there is at most one element $t$ such that $right(t) = c$ and $top(t) = c'$.
\end{itemize}
\end{definition}

\begin{figure}[t]
	\begin{center}
		\begin{tikzpicture}[scale=0.7]
		\draw[fill=lightgray] (0, 3) rectangle (3, 0);
						\foreach \x in {0.5, 1.5} 
						\foreach \y in {0.5, 2.5}{
							\draw[->] (\x, \y) -- (\x+1, \y);
						}
						\foreach \x in {0.5, 1.5} 
						\foreach \y in {1.5}{
							\draw[->] (\x+1, \y) -- (\x, \y);
						}		
						\draw[->] (2.5, 2.5) edge[in=0, out=0] (2.5, 1.5);
						\draw[->] (0.5, 1.5) edge[in=180, out=180] (0.5, 0.5);
		\tile{0}{2}{\tilewhite}{\tilewhite}{\tilegreen}{\tilered}
		%		\tile{1}{2}{\tilegreen}{\tilered}{\tilegreen}{\tileyellow}
		%		\tile{2}{2}{\tilegreen}{\tilered}{\tilegreen}{\tileyellow}
		
		\end{tikzpicture}
	\end{center}
	\caption{Filling a rectangle in the Boustrophedon order.\label{figure:boustrophedon}}
\end{figure}

In other words, when the set of tiles is deterministic, it means that we can deterministically complete a tiling -- if it exists -- in the Boustrophedon order, as shown in Figure~\ref{figure:boustrophedon}.
Note that the fact that $\settilingtypes$ is deterministic can be tested in log-space in the size of $\settilingtypes$. 
We define $\deterministictilingproblem{h(n), w(n)}$
 the restriction of $\tilingproblem{h(n), w(n)}$ to inputs in which $\settilingtypes$ is deterministic.

\subsection{Complexity Results}

\begin{theorem}
	\label{theorem:detPTIME}
	$\deterministictilingproblem{n, n}$ is \PTIME-complete.
\end{theorem}

\begin{proof}
	We design a deterministic algorithm that decides $\deterministictilingproblem{n, n}$ in polynomial-time as follows: it tries to construct the tiling of the $n\times n$-rectangle without backtrack, in the Boustrophedon order, since  $\settilingtypes$ is deterministic.
	
	Let $A$ be a problem in $\PTIME$. Let $\turingmachine$ be a Turing machine that decides $A$	 in polynomial-time. The reduction is as in the proof of Theorem~\ref{theorem:tilingnnNP-complete}, since $\settilingtypes_{\turingmachine, \word}$ is deterministic.

\end{proof}

\begin{theorem}
	\label{theorem:detPSPACE}
	$\deterministictilingproblem{2^n, n}$ and 	$\deterministictilingproblem{*, n}$  is \PSPACE-complete.
\end{theorem}

\begin{theorem}
	\label{theorem:detkexp}
	$\deterministictilingproblem{\expo k n, \expo k n}$ is $k$\EXPTIME-complete.
\end{theorem}

\begin{theorem}
	\label{theorem:detkexpspace}
	Let $k \geq 1$. 
	$\deterministictilingproblem{\expo {k+1} n, \expo {k} n}$ and 	$\deterministictilingproblem{*, \expo {k} n}$  is $\kEXPSPACE[k]$-complete.
\end{theorem}

\section{Existence of Tilings for Classes below \PTIME}
\label{section:existencebelowP}

\subsection{\FO}

\FO is the class of decision problems such that the set of positive instances is described by a logical formula of first-order logic~(see the book on descriptive complexity by Immerman, \cite{immerman2012descriptive}).
\begin{theorem}
	\label{theorem:fo}
	Let $k, \ell$ be two constants.
	$\tilingproblem{k, \ell}$ is in FO.
\end{theorem}

\begin{proof}
	For instance, $\tilingproblem{2, 2}$ corresponds to the first-order formula $$\exists t_1, t_2, t_3, t_4, H(t_1, t_2) \land H(t_3, t_4) \land V(t_1, t_3) \land V(t_2, t_4)$$
	
	where predicates $H$ and $V$ encode respectively the horizontal and vertical constraints.
\end{proof}

\subsection{\NLOGSPACE}

In this section, the width of rectangles is 1, so left- and right- colors are irrelevant.

\begin{theorem}
	$\tilingproblem{n, 1}$ and 	$\tilingproblem{*, 1}$ are \NLOGSPACE-complete.
	\label{theorem:tiling1n}
\end{theorem}

\begin{proof}
	The following non-deterministic algorithm decides $\tilingproblem{n, 1}$ in log-space.
\begin{center}
	\begin{minipage}{10cm}
		\begin{algo}
		\begin{algoblocprocedure}{algo($\settilingtypes, n, \tileseed$)}
			
			Let $\atile := \tileseed$
			
			\begin{algoblocfor}{$k$ := 1 to $n$}

				\algochoose $\atile' \in \settilingtypes$ such that  $\verticalconstraint{\atile}{\atile'}$

			\end{algoblocfor}
			
			\algoaccept
		\end{algoblocprocedure}
	\end{algo}
	
	\end{minipage}
\end{center}

	We reduce in log-space the reachability problem ($s-t$-connectivity problem) to $\tilingproblem{n, 1}$ as follows. Let $\tuple{G, s, t}$ be an instance of the $s-t$-connectivity problem. We construct in log-space the following instance of $\tilingproblem{n, 1}$:
	\begin{itemize}
		\item $n$ is 2 + the number of nodes in $G$;
		\item $\settilingtypes$ contains exactly the tiles \tikz{\tiletext {0} {0} {} {} {} {$s$}}, \tikz{\tiletext {0} {0} {} {$s$} {} {$s$}}, \tikz{\tiletext {0} {0} {} {$t$} {} {}},  \tikz{\tiletext {0} {0} {} {$u$} {} {$v$}} whenever there is an edge $(u, v)$ in $G$;
		% \tikz{\tiletext {0} {0} {} {$s,i$} {} {$s,i+1$}} for $i < \frac n 2$, \tikz{\tiletext {0} {0} {} {$t,i$} {} {$t,i+1$}}, for $i > \frac n 2$ \tikz{\tiletext {0} {0} {} {$u,i$} {} {$v,i+1$}} for $i \in \set{1, \dots, n-2}$ whenever there is an edge $(u, v)$ in $G$;
		
		\item the seed is  \tikz{\tiletext {0} {0} {} {} {} {$s$}}.

	\end{itemize}

There is a path from $s$ to $t$ in $G$ iff we can tile the $n\times 1$-rectangle. %Note that $n$ is large enough so that if there is a tile, it must end with $t, n-1$ at the bottom.
\end{proof}

%
%
%\begin{theorem}
%	The variant of 
%	$\tilingproblem{n, 1}$ without seed is \NLOGSPACE-complete.
%\end{theorem}
%
%
%\begin{proof}
%	The proof is similar that the one of Theorem~\ref{theorem:tiling1n}. For the upper bound, we non deterministically choose $\atile \in \settilingtypes$ instead of taking $\atile := \tileseed$. For the lower bound, the reduction is similar except that the seed is not given (for sure, we should have $s, 0$ at the top).
%	
%	
%\end{proof}

In the same way (it refines Th. 2.7 in \cite{DBLP:journals/tcs/Etzion-PetruschkaHM94}):

\begin{theorem}
	For all constants $k$ (not part of the input), the variant of \linebreak[4] $\tilingproblem{n, k}$ without seed is \NLOGSPACE-complete.
\end{theorem}

\subsection{Rotating tiles and \LOGSPACE}

\newcommand{\rotatetilingproblem}[1]{rot\tilingproblem{#1}}

In order to capture \LOGSPACE, we introduce tile types that can be rotated by 180 degrees. We define $\rotatetilingproblem{h(n), w(n)}$ the restriction of $\tilingproblem{h(n), w(n)}$ to inputs in which $\settilingtypes$ is such that:

\begin{itemize}
	\item if $\tuple{left(t), up(t), right(t), down(t)} \in \settilingtypes$ then
	  $\tuple{right(t), bottom(t), left(t), up(t)} \in \settilingtypes$.
\end{itemize}

%
%\begin{proposition}
%	$\rotatetilingproblem{1 \times *, top, bottom, rot}$ is in \LOGSPACE.
%\end{proposition}
%
%\begin{proof}
%	We reduce 	$\tilingproblem{1 \times *, top, bottom, rot}$ in log-space to the the $s-t$-Uconnectivity problem, which is in L:
%	\begin{itemize}
%		\item the undirected graph whose nodes are the tiles in $\settilingtypes$ and we add an edge $t -- t'$ when $\verticalconstraint{\atile}{\atile'}$;
%		\item the source is $bottom$; the target is $top$.
%	\end{itemize} 
%\end{proof}
%
%

\begin{theorem}
	\label{theorem:LOGSPACE}
	$\rotatetilingproblem{n, 1}$ and 	$\rotatetilingproblem{*, 1}$ are \LOGSPACE-complete, w.r.t. \FO-reductions\footnote{Note that \LOGSPACE is a too small class for log-space reductions to be meaningful.}.
\end{theorem}

\begin{proof}
	We reduce 	$\rotatetilingproblem{n, 1}$ in log-space to the reachability problem in undirected graphs, which is in \LOGSPACE \cite{DBLP:journals/jacm/Reingold08}:
	\begin{itemize}
		\item the nodes of the undirected graph are a source, $n$ copies of $\settilingtypes$, a target;
		\item We add edges from the source to all tiles in the first copy of $\settilingtypes$ if its top is white; we add edges between $t$ of the $i^{\text{th}}$ copy of  $\settilingtypes$ and $t'$ of the $i+1^{\text{th}}$ copy of  $\settilingtypes$ whenever $\verticalconstraint{\atile}{\atile'}$; we add edges between any tile in the $n^{\text{th}}$ copy of $\settilingtypes$ whose bottom is white  and the target.
	\end{itemize}

The reduction given in the proof of~Theorem~\ref{theorem:tiling1n} is also a reduction from the reachability problem in undirected graphs to $\rotatetilingproblem{n, 1}$. This reduction is a \FO-reduction (you can define the set of tiles via first-order formulas).
\end{proof}

%\todo{	$\tilingproblem{\log n \times n, rot}$ is in $\LOGSPACE^2$?}
%
%\todo{	$\tilingproblem{\log n \times \log n, rot}$ in $\LOGSPACE^3$?}
%
%
%\todo{	$\tilingproblem{\log n \times \log n, rot}$ in $\LOGSPACE^3$?}

%
%
%\subsection{Games}
%
%\begin{theorem}
%$\tilingproblem{n, 1, (\exists \forall)^*}$ is \PTIME-complete.
%\end{theorem}
%
%\begin{proof}
%	
%\end{proof}

\section{Conclusion}
\label{section:conclusion}

Table~\ref{table:complexity} sums up the main complexity results for tiling. There are many research avenues, to name a few:
\begin{itemize}
	\item how to define tiling problems with imperfect information in the spirit of~\cite{DBLP:conf/focs/PetersonR79}?
	\item how to define parameterized tiling problems in the spirit of parameterized complexity?
	\item how to get rid off the seed in some of tiling problems and/or border constraints?
	\item what are the connections between tilings and other classes such as AC (alternating circuits), NC (Nick's class), the Boolean hierarchy?
%	\item is 	$\rotatetilingproblem{n, 1}$ \LOGSPACE-complete, by using \FO-reductions?
	
	\item is $\tilingproblem{cst, cst}$ \FO-complete in some sense?
	\item could we have a more natural definition of \emph{deterministic} tilings?
\end{itemize}

\paragraph{Acknowledgments.} I would like to thank Sophie Pinchinat for pointing out the \linebreak[4]class \AEXPTIMEpol. Thanks to Stephane Demri for the discussions about the  \linebreak[4]class \AEXPTIMEpol. Thanks to Sasha Rubin and Tristan Charrier for having given me the motivation to write this note. Especially thanks to Sasha Rubin for his comments on a previous version of that document. Thanks to Florian Beau for the discussion about Tetravex.

\begin{table}[t]
	
	\begin{center}
		\begin{tabular}{lllll}
		in \FO & $\tilingproblem{cst, cst}$  & Th.~\ref{theorem:fo} \\
		\hline
		\LOGSPACE-complete & $\rotatetilingproblem{n, 1}$ & Th.~\ref{theorem:LOGSPACE} \\
		& $\rotatetilingproblem{*, 1}$ & \\
		\hline
		\NLOGSPACE-complete & $\tilingproblem{n, 1}$ & Th.~\ref{theorem:tiling1n} \\
		 & $\tilingproblem{*, 1}$ & \\
		\hline 
		\PTIME-complete & 
		$\deterministictilingproblem{n, n}$ & 	Th.~\ref{theorem:detPTIME} \\
		\hline
		\NP-complete & $\tilingproblem{n, n}$ & Th.~\ref{theorem:tilingnnNP-complete} \\
		\hline
		$\Sigma_k^P$-complete & $\tilingproblem{kn, n, \exists^n(\forall^n\exists^n)^{k-1}}$  & Th.~\ref{theorem:PH} \\
$\Pi_k^P$-complete	& 	$\tilingproblem{kn, n,  \forall^n(\exists^n\forall^n)^{k-1}}$  & Th.~\ref{theorem:PH} \\
\hline
		\PSPACE-complete & 
		$\tilingproblem{ 2^n, n}$ & Th.~\ref{theorem:PSPACE} \\
		& $\tilingproblem{*, n}$ & \\
		&  $\deterministictilingproblem{2^n, n}$ &
		Th.~\ref{theorem:detPSPACE} \\
		&  	$\deterministictilingproblem{*, n}$ \\
		& $\tilingproblem{n, n, (\exists\forall)^*}$ & Th.~\ref{theorem:tilingnnPSPACEgame} \\
		\hline
		\AEXPTIMEpol-complete & $\tilingproblem{2 \times n \times 2^n, 2^n, (\exists^{2^n}\forall^{2^n})^*}$ & Th.~\ref{theorem:aexppol} \\
		\hline
		$k$\EXPTIME-complete & $\tilingproblem{\expo k n, \expo {k-1} n,  (\exists\forall)^*}$ & Th.~\ref{theorem:existencekexp} \\
		& $\tilingproblem{*, \expo {k-1} n,  (\exists\forall)^*}$ & \\
		& $\deterministictilingproblem{\expo k n, \expo k n}$ & Th.~\ref{theorem:detkexp} \\
		\hline
		$k$\NEXPTIME-complete & $\tilingproblem{\expo k n, \expo k n}$ & Th.~\ref{theorem:knexptime} \\
		\hline
		$k$\EXPSPACE-complete & $\tilingproblem{\expo {k+1} n, \expo {k} n}$ & Th.~\ref{theorem:kEXPSPACE} \\
		 &  	$\tilingproblem{*, \expo {k} n}$ & \\
		 		&  $\deterministictilingproblem{\expo {k+1} n, \expo {k} n}$ & Th.~\ref{theorem:detkexpspace}  \\
		 & $\deterministictilingproblem{*, \expo {k} n}$ & \\
		 & $\tilingproblem{\expo k n, \expo k n, (\exists\forall)^*}$ & Th.~\ref{theorem:gameskexpspace} \\
		\end{tabular}
	\end{center}
	\caption{Complexities of tiling ($n$ is in \textbf{unary}).\label{table:complexity}}
\end{table}

\newpage
\bibliographystyle{plain}
\bibliography{biblio}

\end{document}

%% file: macros.tex
\newcommand{\atmset}{\ensuremath{\mathit{AP}}\xspace}
\newcommand{\AP}{{\atmset}}

\protected\def\MSO{\ifmmode \mbox{\sc MSO} \else {\sc MSO}\xspace\fi}
\protected\def\FO{\ifmmode \mbox{\sc FO} \else {\sc FO}\xspace\fi}
\protected\def\MMSO{\ifmmode {\mbox{\sc M}\MSO} \else {{\sc M}\MSO}\xspace\fi}
\protected\def\MFO{\ifmmode \mbox{\sc MFO} \else {\sc MFO}\xspace\fi}

\newcommand{\modelM}{\mathcal M}
\newcommand{\eventmodel}{\mathcal E}

\definecolor{TODO_COLOR}{rgb}{1,0.5,0.5}

\newcommand\AEXPTIMEpol{A{$_{pol}$}EXPTIME}

\newcommand{\set}[1]{{\{#1\}}}

\renewcommand{\phi}{\varphi}

\newcommand{\tuple}[1]{\langle #1 \rangle}

\newcommand{\setN}{\mathbb{N}}

\newcommand{\ensN}{\setN}

\newtheorem{theorem}{Theorem}

\newtheorem{definition}{Definition}

\newenvironment{proof}{
 \textsc{Proof.}
   {~\\}
   \normalfont
   \indent
 }{$\blacksquare$}

\newcommand{\algofunction}{\textbf{function }}

\newcommand{\algoprocedure}{\textbf{procedure }}

\newcommand{\algoendfunction}{\textbf{endFunction }}

\newcommand{\algofor}{\textbf{for }}
\newcommand{\algodo}{\textbf{do }}

\definecolor{algocommentbackgroundcolor}{rgb}{1,1,0.5}

\newcommand{\algowhile}{\textbf{while }}

\newcommand{\algoif}{\textbf{if }}
\newcommand{\algothen}{\textbf{then }}

\newcommand{\algoendif}{\textbf{endIf }}

\newcommand{\algomatch}{\textbf{match }}

\newcommand{\algocase}{\textbf{case }}

\newcommand{\algochoose}{\textbf{choose }}

\newcommand{\algoaccept}{\textbf{accept }}

\newlength{\algoindentlongueur}
\newcommand{\algoindent}{\hspace*{\algoindentlongueur}}

\newlength{\algoindentavantvrulelongueur}
\newcommand{\algoindentavantvrule}{\hspace*{\algoindentavantvrulelongueur}}

\newlength{\dummy}

\newsavebox{\frameminipageboiteavecunnomsuperlongdesortequonnepuissepaslereutiliser}
\newenvironment{frameminipage}[2][c]{%
\begin{lrbox}{\frameminipageboiteavecunnomsuperlongdesortequonnepuissepaslereutiliser}%
\begin{minipage}[#1]{#2}%
} {%
\end{minipage}%
\end{lrbox}%
\framebox{\usebox{\frameminipageboiteavecunnomsuperlongdesortequonnepuissepaslereutiliser}}%
}

\newenvironment{algo} {
  \begin{frameminipage}{\linewidth}
} {
  \end{frameminipage}
}

\newenvironment{algobloc}{\setlength{\dummy}{\linewidth}\addtolength{\dummy}{- \algoindentlongueur}\addtolength{\dummy}{- \algoindentavantvrulelongueur}\algoindentavantvrule\vrule\algoindent\begin{minipage}{\dummy}}{\end{minipage}}

\newenvironment{algoblocprocedure}[1]
{\algoprocedure #1 \\  \begin{algobloc}}
	{\end{algobloc}}

\newenvironment{algoblocfor}[1]
{\algofor #1 \algodo \\  \begin{algobloc}}
{\end{algobloc}}

\tikzstyle{zoneprogramcounter} = [fill=gray!40, draw=none]
\tikzstyle{zonedatacounter} = [fill=gray!40, draw=none, decorate, decoration={snake, amplitude=1pt, segment length=4pt}]
\tikzstyle{portion} = [densely dotted, shade, top color = white, bottom color = gray!40]
\tikzstyle{portiondata} = [densely dotted, shade, top color = white, bottom color = gray!40, decorate, decoration={snake, amplitude=1pt, segment length=4pt}]
\tikzset{
	double arrow/.style args={#1 colored by #2 and #3}{
		-stealth,line width=#1,#2, % first arrow
		postaction={draw,-triangle 90 cap,#3,line width=(#1)-4pt,
			shorten <=4pt,shorten >=4}, % second arrow
	}
}
\tikzstyle{tapearrow} = [double arrow=10pt colored by black!50!white and black!20!white, -triangle 90 cap, fill = white]%
\tikzstyle{execarrow} = [line width=1pt,->]%

\tikzstyle{world}=[inner sep=0.5mm]
\tikzstyle{event}=[fill=gray!30, inner sep=0.5mm]
\tikzstyle{realworldarrowfromleft} = [initial left, initial text={}]
\tikzstyle{realworldarrowfromright} = [initial right, initial text={}]

	\newcommand{\stateaccept}{{accept}}

\protected\def\DTIME{\ifmmode \mbox{\sc Dtime} \else {\sc Dtime}\xspace\fi}
\protected\def\NTIME{\ifmmode \mbox{\sc Ntime} \else {\sc Ntime}\xspace\fi}
\protected\def\DSPACE{\ifmmode \mbox{\sc Dspace} \else {\sc Dspace}\xspace\fi}
\protected\def\NSPACE{\ifmmode \mbox{\sc Nspace} \else {\sc Nspace}\xspace\fi}
\protected\def\ASPACE{\ifmmode \mbox{\sc Aspace} \else {\sc Aspace}\xspace\fi}
\protected\def\LOGSPACE{\ifmmode \mbox{\sc LOGSPACE} \else {\sc LOGSPACE}\xspace\fi}
\protected\def\NLOGSPACE{\ifmmode \mbox{\sc NLOGSPACE} \else {\sc NLOGSPACE}\xspace\fi}
\protected\def\NP{\ifmmode \mbox{\sc NP} \else {\sc NP}\xspace\fi}
\protected\def\AP{\ifmmode \mbox{\sc AP} \else {\sc AP}\xspace\fi}
\protected\def\coNP{\ifmmode \mbox{\sc coNP} \else {\sc coNP}\xspace\fi}
\protected\def\NPSPACE{\ifmmode \mbox{\sc NPspace} \else {\sc NPspace}\xspace\fi}
\protected\def\APSPACE{\ifmmode \mbox{\sc APspace} \else {\sc APspace}\xspace\fi}
\protected\def\PSPACE{\ifmmode \mbox{\sc Pspace} \else {\sc Pspace}\xspace\fi}
\protected\def\EXPSPACE{\ifmmode \mbox{\sc Expspace} \else {\sc Expspace}\xspace\fi}
\protected\def\TWOEXPSPACE{\ifmmode \mbox{\sc 2Expspace} \else {\sc 2Expspace}\xspace\fi}
\protected\def\PTIME{\ifmmode \mbox{\sc P} \else {\sc P}\xspace\fi}
\protected\def\NPTIME{\ifmmode \mbox{\sc NP} \else {\sc NP}\xspace\fi}
\protected\def\EXPTIME{\ifmmode \mbox{\sc Exptime} \else {\sc Exptime}\xspace\fi}
\protected\def\AEXPTIME{\ifmmode \mbox{\sc Aexptime} \else {\sc Aexptime}\xspace\fi}
\protected\def\NEXPTIME{\ifmmode \mbox{\sc NExptime} \else {\sc NExptime}\xspace\fi}
\protected\def\2EXPTIME{\ifmmode \mbox{\sc 2-Exptime} \else {\sc
		2-Exptime}\xspace\fi}
\DeclareRobustCommand{\kEXPTIME}[1][k]{\ifmmode \mbox{\sc $#1$-Exptime}
	\else {\sc $#1$-Exptime}\xspace\fi}
\DeclareRobustCommand{\kNEXPTIME}[1][k]{\ifmmode \mbox{\sc $#1$-NExptime}
	\else {\sc $#1$-NExptime}\xspace\fi}
\DeclareRobustCommand{\kAEXPTIME}[1][k]{\ifmmode \mbox{\sc $#1$-AExptime}
	\else {\sc $#1$-AExptime}\xspace\fi}
\DeclareRobustCommand{\kEXPSPACE}[1][k]{\ifmmode \mbox{\sc $#1$-Expspace}
	\else {\sc $#1$-Expspace}\xspace\fi}
\DeclareRobustCommand{\kAEXPSPACE}[1][k]{\ifmmode \mbox{\sc $#1$-AExpspace}
	\else {\sc $#1$-AExpspace}\xspace\fi}
\protected\def\ELEMENTARY{\ifmmode \mbox{\sc Elementary} \else {\sc Elementary}\xspace\fi}
\protected\def\AEXPpol{\ifmmode \mbox{{\sc A}_{\text{pol}}\EXPTIME} \else
	{\sc A}$_{\text{pol}}$\EXPTIME\fi}
\protected\def\APTIME{\ifmmode \mbox{\sc Aptime} \else {\sc Aptime}\xspace\fi}
\protected\def\AEXPSPACE{\ifmmode \mbox{\sc Aexpspace} \else {\sc Aexpspace}\xspace\fi}

\newcommand{\gloups}[1]{\bigcirc}

\tikzstyle{cell} = [draw,minimum height=5mm,minimum width=5mm]
\tikzset{
	cellcolor/.cd,
	0/.style={fill=gray!20!white},
	1/.style={fill=yellow!20!white}
}
\tikzstyle{cellalive} = [fill=yellow!20!white]

\newcommand{\word}{w}

\newcommand\listeventsempty\epsilon

\NewDocumentCommand{\update}{O{}}{\modelM \otimes \eventmodel^{#1}}
\NewDocumentCommand{\DELstruct}{O{*}}{\modelM \eventmodel^{#1}}

\newcommand{\tile}[6]{
	\draw[fill=#3] (#1,#2) -- (#1,#2+1) -- (#1+0.5, #2+0.5) -- (#1, #2);
	\draw[fill=#4] (#1,#2+1)-- (#1+1,#2+1) -- (#1 +0.5, #2+0.5) -- (#1, #2+1);
	\draw[fill=#5] (#1+1,#2)-- (#1+1,#2+1) -- (#1 +0.5, #2+0.5) -- (#1 +1, #2);
	\draw[fill=#6] (#1,#2)-- (#1 +1,#2) -- (#1 +0.5, #2+0.5) -- (#1, #2);
	\draw[fill=black] (#1+0.5, #2+0.5) -- (#1+0.6, #2+0.4) -- (#1+0.4, #2+0.4) -- cycle;
}

\newcommand{\tiletext}[6]{
	\tile {#1} {#2} {none} {none} {none} {none}
	\node at (#1+0.25, #2+0.5) {#3}; %left
	\node at (#1+0.5, #2+0.8) {#4}; %up
	\node at (#1+0.75, #2+0.5) {#5}; %right
	\node at (#1+0.5, #2+0.1) {#6}; %bottom
}

\newcommand{\tilered}{red!80}
\newcommand{\tileyellow}{yellow!50}
\newcommand{\tilegreen}{green!70}

\newcommand{\tilewhite}{white}

\definecolor{C}{rgb}{1,1,1}
\definecolor{C'}{rgb}{1,1,0.9}
\definecolor{C'}{rgb}{1,1,0.9}
\definecolor{C }{rgb}{0.8,0.8,0.8}

\definecolor{Cs}{rgb}{0.95,1,0.5}
\definecolor{Ct}{rgb}{0.95,1,0.4}

\definecolor{Cu}{rgb}{0.95,1,0.5}
\definecolor{Cv}{rgb}{0.95,1,0.4}

\definecolor{Ca}{rgb}{0.95,1,0.5}
\definecolor{Ca'}{rgb}{0.95,1,0.5}
\definecolor{Cb}{rgb}{1,0.9,0.5}
\definecolor{Cb'}{rgb}{1,0.9,0.5}

\definecolor{Cq'-}{rgb}{0.5,0.5,1}
\definecolor{Cq'}{rgb}{0.5,0.5,1}
\definecolor{C-q_0}{rgb}{0.5,0.5,1}
\definecolor{C-q_1}{rgb}{0.5,0.4,1}
\definecolor{Cq_0}{rgb}{0.5,0.5,1}
\definecolor{Cq_1}{rgb}{0.5,0.4,1}
\definecolor{Cq_1'}{rgb}{0.5,0.4,1}
\definecolor{Cq_2'}{rgb}{0.5,0.4,1}
\definecolor{C-q'}{rgb}{0.5,0.8,0.5}
\definecolor{Cq}{rgb}{0.5,0.8,0.5}
\definecolor{Cq_f-}{rgb}{0.4,0.9,0.5}

\definecolor{C2}{rgb}{1,0.6,0.5}
\definecolor{C3}{rgb}{1,0.6,0.5}
\definecolor{C|w|}{rgb}{1,0.6,0.5}
\definecolor{C|w|{^+}1}{rgb}{1,0.6,0.5}
\definecolor{Cw_{|w|}'}{rgb}{0.95,1,0.5}
\definecolor{Cq_0',w_1'}{rgb}{1,0.6,0.5}
\definecolor{Cw_2'}{rgb}{0.95,1,0.5}
\definecolor{Cq_0,}{rgb}{1,0.6,0.4}
\definecolor{Cq_2,}{rgb}{0.8,0.9,0.4}
\definecolor{Cq_2,a}{rgb}{0.8,0.9,0.4}
\definecolor{Cq_f,a}{rgb}{0.7,0.9,0.4}
\definecolor{Cq_0,a}{rgb}{1,0.6,0.5}
\definecolor{Cq_0,b}{rgb}{1,0.5,0.5}
\definecolor{Cq_0',b'}{rgb}{1,0.5,0.5}
\definecolor{Cq_1,a}{rgb}{1,0.5,0.4}
\definecolor{Cq_1,b}{rgb}{1,0.4,0.4}
\definecolor{Cq_1',b'}{rgb}{1,0.4,0.4}
\definecolor{Cq_2',b'}{rgb}{1,0.4,0.4}
\definecolor{Cq,a}{rgb}{1,0.5,0.5}
\definecolor{Cq,b}{rgb}{1,0.5,0.5}
\definecolor{Cq',a}{rgb}{1,0.5,0.5}
\definecolor{Cq',a'}{rgb}{1,0.5,0.5}
\definecolor{Cq',b}{rgb}{1,0.5,0.5}
\definecolor{Cq',b'}{rgb}{1,0.5,0.5}
\definecolor{Cq_0,1}{rgb}{0.5,0.5,0.5}

\newcommand{\tiletextcolor}[6]{
	\def\templeft{C#3}	
	\def\tempup{C#4}
	\def\tempright{C#5}
	\def\tempdown{C#6}
	\tile{#1}{#2}{\templeft}{\tempup}{\tempright}{\tempdown}
	\tiletext {#1} {#2} {$#3$} {$#4$} {$#5$} {$#6$}}

\newcommand{\tilecopytoprime}[3]{
	\tiletextcolor {#1} {#2} {} {#3} {} {#3'}}

\newcommand{\tilecopyfromprime}[3]{
	\tiletextcolor {#1} {#2} {} {#3'} {} {#3}}

\newcommand{\tilegotoright}[6]{
	\tiletextcolor {#1} {#2} {} {#3,#4} {#6} {#5}}

\newcommand{\tilerecfromright}[4]{
	\tiletextcolor {#1} {#2} {#3} {#4} {} {#3,#4}}

%% file: document.bbl
\begin{thebibliography}{10}

\bibitem{LICS2019DEMRI}
Bartosz Bednarczyk and St\'ephane Demri.
\newblock Why propositional quantification makes modallogics on trees robustly
  hard.
\newblock In {\em LICS 2019}.

\bibitem{DBLP:journals/corr/abs-1709-02094}
Laura Bozzelli, Alberto Molinari, Angelo Montanari, and Adriano Peron.
\newblock On the complexity of model checking for syntactically maximal
  fragments of the interval temporal logic {HS} with regular expressions.
\newblock In {\em Proceedings Eighth International Symposium on Games,
  Automata, Logics and Formal Verification, GandALF 2017, Roma, Italy, 20-22
  September 2017.}, pages 31--45, 2017.

\bibitem{DBLP:conf/jelia/BozzelliDP12}
Laura Bozzelli, Hans van Ditmarsch, and Sophie Pinchinat.
\newblock The complexity of one-agent refinement modal logic.
\newblock In {\em Logics in Artificial Intelligence - 13th European Conference,
  {JELIA} 2012, Toulouse, France, September 26-28, 2012. Proceedings}, pages
  120--133, 2012.

\bibitem{chandra1976alternation}
A.K. Chandra and L.J. Stockmeyer.
\newblock {Alternation}.
\newblock In {\em 17th annual symposium on Foundations of Computer Science},
  pages 98--108. IEEE, 1976.

\bibitem{DBLP:conf/atal/CharrierS15}
Tristan Charrier and Fran{\c{c}}ois Schwarzentruber.
\newblock Arbitrary public announcement logic with mental programs.
\newblock In {\em Proceedings of the 2015 International Conference on
  Autonomous Agents and Multiagent Systems, {AAMAS} 2015, Istanbul, Turkey, May
  4-8, 2015}, pages 1471--1479, 2015.

\bibitem{DBLP:journals/jcss/Chlebus86}
Bogdan~S. Chlebus.
\newblock Domino-tiling games.
\newblock {\em J. Comput. Syst. Sci.}, 32(3):374--392, 1986.

\bibitem{DBLP:journals/tcs/Etzion-PetruschkaHM94}
Yael Etzion{-}Petruschka, David Harel, and Dale Myers.
\newblock On the solvability of domino snake problems.
\newblock {\em Theor. Comput. Sci.}, 131(2):243--269, 1994.

\bibitem{DBLP:journals/siamcomp/Gradel90}
Erich Gr{\"{a}}del.
\newblock Domino games and complexity.
\newblock {\em {SIAM} J. Comput.}, 19(5):787--804, 1990.

\bibitem{immerman2012descriptive}
Neil Immerman.
\newblock {\em Descriptive complexity}.
\newblock Springer Science \& Business Media, 2012.

\bibitem{Papadimitriou}
Christos~H Papadimitriou.
\newblock {\em Computational complexity}.
\newblock John Wiley and Sons Ltd., 2003.

\bibitem{DBLP:conf/focs/PetersonR79}
Gary~L. Peterson and John~H. Reif.
\newblock Multiple-person alternation.
\newblock In {\em 20th Annual Symposium on Foundations of Computer Science, San
  Juan, Puerto Rico, 29-31 October 1979}, pages 348--363, 1979.

\bibitem{DBLP:journals/jacm/Reingold08}
Omer Reingold.
\newblock Undirected connectivity in log-space.
\newblock {\em J. {ACM}}, 55(4):17:1--17:24, 2008.

\bibitem{DBLP:journals/ipl/TakenagaW06}
Yasuhiko Takenaga and Toby Walsh.
\newblock Tetravex is np-complete.
\newblock {\em Inf. Process. Lett.}, 99(5):171--174, 2006.

\bibitem{DBLP:conf/jelia/Turner02}
Hudson Turner.
\newblock Polynomial-length planning spans the polynomial hierarchy.
\newblock In {\em Logics in Artificial Intelligence, European Conference,
  {JELIA} 2002, Cosenza, Italy, September, 23-26, Proceedings}, pages 111--124,
  2002.

\bibitem{van1997convenience}
Peter van Emde~Boas et~al.
\newblock The convenience of tilings.
\newblock {\em Lecture Notes in Pure and Applied Mathematics}, pages 331--363,
  1997.

\bibitem{DBLP:journals/cacm/Wang60}
Hao Wang.
\newblock Proving theorems by pattern recognition {I}.
\newblock {\em Commun. {ACM}}, 3(4):220--234, 1960.

\end{thebibliography}
